\documentclass{article}
 
\pdfoutput=1 
 
\usepackage{microtype}
\usepackage{amsmath, amsfonts, amsthm}
\usepackage{enumerate}
\usepackage[linesnumbered,boxed,ruled,vlined]{algorithm2e}
\usepackage[T1]{fontenc}
\usepackage[utf8]{inputenc}
\usepackage{authblk}
\usepackage[titletoc]{appendix}

\theoremstyle{plain}
\newtheorem{thm}{Theorem}[section]
\newtheorem{lem}[thm]{Lemma}

\theoremstyle{definition}
\newtheorem{defn}{Definition}[section]

\theoremstyle{remark}

\newtheorem*{note}{Note}

\title{Improved distance sensitivity oracles via tree partitioning}
\author[1]{Ran Duan}
\author[2]{Tianyi Zhang}
\affil[1]{Institute for Interdisciplinary Information Science, Tsinghua University, China\\
	\texttt{duanran@mail.tsinghua.edu.cn}}
\affil[2]{Institute for Interdisciplinary Information Science, Tsinghua University, China\\
	\texttt{zhangty12@mails.tsinghua.edu.cn}}
\date{}

\begin{document}

\maketitle

\begin{abstract}
We introduce an improved structure of \textsl{distance sensitivity oracle} (DSO). The task is to pre-process a non-negatively weighted graph so that a data structure can quickly answer replacement path length for every triple of source, terminal and failed vertex. The previous best algorithm constructs in time \footnote{$\tilde{O}(\cdot)$ suppresses poly-logarithmic factors.}$\tilde{O}(mn)$ a distance sensitivity oracle of size $O(n^2\log n)$ that processes queries in $O(1)$ time. As an improvement, our oracle takes up $O(n^2)$ space, while preserving $O(1)$ query efficiency and $\tilde{O}(mn)$ preprocessing time. One should notice that space complexity and query time of our novel data structure are asymptotically optimal.
\end{abstract}

\section{Introduction}
Let $G = (V, E)$ be a directed graph with $n$ vertices and $m$ edges, and edge weights $\omega: E\rightarrow R^+\cup \{0\}$. We wish to pre-process the graph and build a distance sensitivity oracle that answers for every triple of $(s, t, f)$ the length of replacement path from $s$ to $t$ that circumvents the assumed failed vertex $f$. We call any data structure that answers such queries \textsl{distance sensitivity oracle} (DSO). 

A close relative of this DSO problem is that we consider edge failures instead of vertex failures. In this paper, we are only concerned with failed vertices, though all our structures can be easily extended (see \cite{Demet08}) to handle edge-failure queries without any loss in space / time efficiency.

Motivation for vertex failures is routing where network nodes occasionally undergo crash failures. In the face of single node-failures, we might not afford to re-compute all-pair shortest paths from scratch before the failed node comes back on-line. A distance sensitivity oracle may lend us help in the sense that re-computing APSP becomes unnecessary. Motivation for edge failures comes from in Vickery pricing \cite{Suri02}, where one wishes to measure, for every pair of source and target as well as a failed edge, by how much the shortest distance would rise if this designated edge were to shut down.

\subsection{Existing algorithms}
The naive approach is that we pre-compute and store the length of all $O(n^3)$ possible replacement path lengths, which incurs intolerable space complexity. \cite{Demet02} proposes the first DSO that occupies only near-quadratic space. More specifically, \cite{Demet02}'s DSO has space complexity $O(n^2\log n)$ and $O(1)$ query time. Space complexity of constant query time DSO has not been improved ever since.

DSO in \cite{Demet02} demands a somewhat high preprocessing time complexity of $O(mn^2)$, which was improved to $\tilde{O}(mn^{\frac{3}{2}})$ in the journal version \cite{Demet08} while the space complexity was blown up to $O(n^{2.5})$. Cubic time preprocessing algorithm was first obtained by \cite{Bernst08} and shortly improved from $\tilde{O}(n^2\sqrt{m})$ to $\tilde{O}(mn)$ in \cite{Bernst09}, while maintaining $O(n^2\log n)$ space and $O(1)$ query time. Note that $O(n^2\log n)$ and $\tilde{O}(mn)$ are basically optimal up to poly-logarithmic factors, as discussed in \cite{Bernst09}. Therefore, surpassing \cite{Bernst09}'s construction time has been deemed hard from then.

Since the publication of \cite{Bernst09}, the community's interest has diverged to seeking truly sub-cubic preprocessing time algorithms. F. Grandoni and V. Williams (\cite{virgi12}) obtained truly sub-cubic preprocessing time bound $O(Mn^{2.88})$, if one should tolerate a sub-linear query time of $O(n^{0.7})$; here all edge weights are assumed to be integers within interval $[-M, M]$.

\subsection{Our contributions}
In this paper we present a DSO construction that improves upon \cite{Bernst09}.

\begin{thm}
For any directed non-negatively weighted graph $G = (V, E, \omega)$, a DSO with $O(n^2)$ space complexity and $O(1)$ query time exists. Also, such DSO can be preprocessed in time $\tilde{O}(mn)$.
\end{thm}

In Bernstein \& Karger's work (\cite{Bernst09}), the space / query time was $O(n^2\log n)$ and $O(1)$. So compared with \cite{Demet08}'s result, our construction shaves off the last $\log n$ factor in the space complexity, leading to a quadratic space consumption, while preserving constant time query efficiency. Plus, our DSO can also be constructed in $\tilde{O}(mn)$ time as in \cite{Bernst09}, which is nearly optimal. One should notice that the space complexity of our DSO is asymptotically \textbf{optimal}. To see that, we argue that merely answering all-pairs distances requires $\Omega(n^2)$ space. Consider when $G$ is a complete directed graph with distinct edge weights drawn from $(1, 2]$. In this case, for any $s, t\in V$, the shortest path from $s$ to $t$ is simply the edge $(s, t)\in E$. Hence, to answer the distance from $s$ to $t$, the data structure can do nothing but store $\omega(s, t)$; otherwise $\omega(s, t)$ can never be precisely known. Thereby, the total storage is at least $|E| = \Omega(n^2)$. For other lower bounds in undirected graphs, see \cite{Patr10}.

The observation comes from Demetrescu's original work. \cite{Demet02}'s construction basically applies the idea of sparse table, where each pair of source and target are associated with $O(\log n)$ sparse table entries, thus resulting in a total storage of $O(n^2\log n)$. The preliminary idea is that we do not store sparse table entries at every source-terminal pair. More specifically, for each single source shortest path tree (SSSP), we designate only a proportion of all tree vertices which are associated with sparse table entries, hence making our data structure even sparser. The set of all designated vertices should be carefully chosen with respect to the topological structure of the SSSP tree.

With the sparser data structure, we can answer queries $(s, t, f)$ when $f$ keeps distance from both of $s$ and $t$. So the bottleneck lies in degenerated cases where $f$ is very close to one of the endpoints $s$ or $t$. For degenerated cases, we can use much smaller sparse tables to cover short paths, resulting in a DSO that only occupies $O(n^2\log\log n)$ space. To obtain an optimum of $O(n^2)$ space complexity, we would need to apply a tabulation technique (more known as the ``Four Russians'' \cite{4Russians}).

Our preprocessing algorithm heavily relies on the notion of \textsl{admissible functions} from \cite{Bernst08}. The idea is that we substitute ``bottleneck'' vertices for intervals in the original construction, without harming the correctness of query algorithm. In this way, we make the DSO easier to initialize.

\subsection{Related work}
There are several generalizations of the distance sensitivity problem. In \cite{Duan09}, the authors considered the scenario where two instead of one vertices could fail. The paper presented a distance oracle with $O(n^2\log^3 n)$ space complexity and $O(\log n)$ query time. As it turned out, things got far more complicated than single vertex-failures, and sadly no non-trivial polynomial preprocessing algorithms were known.

There are works (e.g., \cite{Thor05}\cite{Demet04}\cite{Demet06}\cite{Chec14}) mainly concerned with dynamically maintaining all pairs shortest paths (APSP). Such data structures can solve distance problems if subsequent failures are cumulative, but update time would be as large as $O(n^{2.75})$.

If one should sacrifice preciseness for space efficiency, one may consider approximate distance oracles for vertex failures. \cite{Basw10} considered approximating the replacement path lengths where a single vertex could crash. In \cite{Chec12}, the authors focused on data structures that approximately answer minimal distance of constrained paths that do not pass through a designated set of failed edges. For fully dynamic approximate APSP, one can refer to \cite{Henz13}; especially, for planar graphs, \cite{Ittai12} may provide useful results. 

There are remotely related problems such as (partially) dynamic single-source shortest paths and reachability. Papers \cite{Henz14a} and \cite{Henz14b} discussed these topics in depth. Some other loosely related work concerns the construction of spanners and distance preservers resilient to one edge / node failure.

\section{Preliminaries}
Suppose we are given a directed graph $G = (V, E, \omega)$ with non-negative edge weights $\omega: E\rightarrow R^+\cup \{0\}$. In this section, we summarize the notations or assumptions that are used throughout this paper. More or less, we inherit the conventions from \cite{Duan09}.

\begin{itemize}
\item Our data structures \& algorithms are implemented on $\Omega(\log n)$-word RAM machines. We will leverage its strength in computing the most significant set bit (\cite{fusion}). Although in previous works on DSO (\cite{Demet04}, \cite{Demet08}, \cite{Bernst08}, \cite{Bernst09}) $\Omega(\log n)$-word RAM model was not explicitly assumed, this assumption was not expendable since their algorithms required memory indexing and computing logarithms in constant time.
\item We call a data structure $\langle f(n), g(n)\rangle$, if its space complexity is at most $f(n)$ and its query time is at most $g(n)$.
\item For each pair of $s, t\in V$, the weighted shortest path from $s$ to $t$ is unique. This assumption is without loss of generality since we can add small perturbations to break ties (e.g., \cite{Demet08}, \cite{Bernst09}).
\item For each pair of $s, t\in V$, let $st$ denote the shortest path from $s$ to $t$.
\item Let $p$ be a simple path. Denote by $\|p\|$ and $|p|$ the weighted and un-weighted length of path $p$.
\item For each $s\in V$, let $T_s$ be the single-source shortest path tree rooted at $s$; let $\widehat{T}_s$ be the single-source shortest path tree rooted at $s$ in the reverse graph $\widehat{G}$ where every directed edge in $G$ is reversed.
\item For each query $(s, t, f)$, we only consider the case when $f$ lies on the path $st$, because verification can be done by checking whether $\|st\| = \|sf\| + \|ft\|$.
\item For each vertex set $A$, let $st\diamond A$ denote the shortest path from $s$ to $t$ that avoids the entire set $A$. For instance, $st\diamond \{f\}$ (abbreviated as $st\diamond f$) denotes the replacement path, and $st\diamond [u, v]$ refers to the shortest path that skips over an entire interval $[u, v]\subseteq st$. Here for $[u, v]$ to be a properly defined interval on $st$, it is required that both of $u$ and $v$ are on $st$, and either $u = v$, or $u$ lying between $s$ and $v$.
\item Let $s\oplus i$ and $s\ominus i$ be the $i^{th}$ vertex after and before $s$ on some path that can be learnt from context.
\end{itemize}

By uniqueness of shortest paths, it is easy to verify that any path of the form $st\diamond f$ or $st\diamond [u, v]$ must diverge from and converge with $st$ \textbf{for once} (\cite{Demet08}), with divergence on path $sf$ and convergence on $ft$. Let $\Delta_{s, t, f}$ and $\nabla_{s, t, f}$ be the vertices at which divergence and convergence take place, respectively. More often than not, we ignore those indices $s, t, f$ and simply write $\Delta, \nabla$.

\section{The sparse table technique \label{old}}
In this section, we review the $\langle O(n^2\log n), O(1)\rangle$ DSO devised in \cite{Demet08}. Basically, its design utilizes the idea of \textsl{sparse table}.

\subsection{Data structure}
For every pair of $s, t\in V$, build the following.

\note Indices $i, j$ are non-negative integers throughout this paper.

\begin{enumerate}[(\romannumeral1)]
\item The values of $\|st\|$ and $|st|$.
\item For every $s, t$ and $2^i < |st|$, the value of $\|st\diamond (s\oplus 2^i)\|$ and $\|st\diamond (t\ominus 2^i)\|$.
\item For every $s, t$ and $2^{i+1} < |st|$, the value of $\|st\diamond [s\oplus 2^i, s\oplus 2^{i+1}]\|$ and $\|st\diamond [t\ominus 2^{i+1}, t\ominus 2^i]\|$.
\item Level ancestor data structures (\cite{Bender03}) for every $T_s$ and $\widehat{T}_s$.
\end{enumerate}

Briefly discuss the space complexity. Evidently (\romannumeral1) takes $O(n^2)$ space; both of (\romannumeral2) and (\romannumeral3) have space complexity $O(n^2\log n)$; for (\romannumeral4), by linear space construction of level ancestor data structure from \cite{Bender03}, each $T_s$ ($\widehat{T}_s$) arouses a storage of $O(n)$, leading to a total space complexity of $O(n^2)$ ranging over all $s$. Thus, the overall space is $O(n^2\log n)$.

\subsection{Query algorithm}
For each query $(s, t, f)$, without loss of generality assume that $|sf|\leq |ft|$. If $|sf|$ or $|ft|$ is a power of 2, then $\|st\diamond f\|$ has already been computed during preprocessing (\romannumeral2). Otherwise, let $i, j$ be the maximum non-negative integers such that $2^i < |sf|, 2^j < |ft|$.

\begin{note}
	Both of $i, j$ can be computed in constant time on $\Omega(\log n)$-word RAM machines: $i$ ($j$) is equal to the index of the most significant set bit of $|sf|-1$ ($|ft|-1$), and according to \cite{fusion} the most significant set bit can be computed in constant time on a $\Omega(\log n)$-word RAM.\\
\end{note}

Consider the following cases.

\begin{enumerate}[(1)]
\item $st\diamond f$ skips the entire $[s\oplus 2^i, s\oplus 2^{i+1}]$ which contains $f$.\\
As we know from (\romannumeral3), $\|st\diamond [s\oplus 2^i, \oplus 2^{i+1}]\|$ can be directly retrieved.
\item $st\diamond f$ passes through $s\oplus 2^i$.\\
By definition of $i$, it must be that $f\ominus 2^i$ lies between $s$ and $s\oplus 2^i$. So in the case that $st\diamond f$ passes through $s\oplus 2^i$, it must also pass through $f\ominus 2^i$. Thereby, $\|st\diamond f\| = \|s(f\ominus 2^i)\| + \|(f\ominus 2^i)t\diamond f\|$; both terms are retrievable from storage in constant time.
\item $st\diamond f$ passes through $s\oplus 2^{i+1}$.\\
By definition of $i, j$, $2^{j+1}\geq |ft|\geq |sf| > 2^i$, which leads to $j\geq i$. Therefore, because $f$ comes after $s\oplus 2^i$, $f\oplus 2^j$ lies in range $[s\oplus 2^{i+1}, t)$. Since $st\diamond f$ passes through $s\oplus 2^{i+1}$, it is guaranteed to also pass through $f\oplus 2^j$. Decomposing $\|st\diamond f\| = \|s(f\oplus 2^j)\diamond f\| + \|(f\oplus 2^j)t\|$, both terms are already pre-computed.
\end{enumerate}

So far we have completed discussion of the $\langle O(n^2\log n), O(1)\rangle$ DSO in \cite{Demet08}.

\section{Admissible functions and the triple path lemma}
Our query algorithms will be frequently using the \textsl{triple path lemma} from \cite{Bernst08}. This lemma is based on the notion of \textsl{admissible functions}.

\begin{defn}[\cite{Bernst08}]
	A function $F_{s, t}^{[u, v]}$ is admissible if $\forall f\in [u, v]$, we have $\|st\diamond [u, v]\|\geq F_{s, t}^{[u, v]}\geq \|st\diamond f\|$.
\end{defn}

\begin{defn}[\cite{Bernst08}]
Two important admissible functions are $\max_{f\in [u, v]}\{\|st\diamond f\|\}$ and $\|st\diamond [u, v]\|$. We call them \textbf{bottleneck} and \textbf{interval} admissible functions, respectively.
\end{defn}

\begin{lem}[The triple path lemma (\cite{Bernst08})]
	Let $[u, v]$ be an interval on $st$, and $f\in [u, v]$ be a vertex. For any admissible function $F_{s, t}^{[u, v]}$, $\|st\diamond f\| = \min\{\|su\| + \|ut\diamond f\|, \|sv\diamond f\| + \|vt\|, F_{s, t}^{[u, v]}\}$.
\end{lem}
\begin{proof}
On the one hand, $\|st\diamond f\|$ is always smaller or equal to $\min\{\|su\| + \|ut\diamond f\|, \|sv\diamond f\| + \|vt\|, F_{s, t}^{[u, v]}\}$. This is because, by definition \textbf{3.1} $F_{s, t}^{[u, v]}\geq \|st\diamond f\|$; also it is easy to see that both of $\|su\| + \|ut\diamond f\|$ and $\|sv\diamond f\| + \|vt\|$ are $\geq \|st\diamond f\|$.

On the other hand, we argue $\|st\diamond f\| \geq \min\{\|su\| + \|ut\diamond f\|, \|sv\diamond f\| + \|vt\|, F_{s, t}^{[u, v]}\}$. If $st\diamond f$ passes through either $u$ or $v$, then $\|st\diamond f\|$ would be equal to either $\|su\| + \|ut\diamond f\|$ or $\|sv\diamond f\| + \|vt\|$, which is $\geq \min\{\|su\| + \|ut\diamond f\|, \|sv\diamond f\| + \|vt\|, F_{s, t}^{[u, v]}\}$. Otherwise $st\diamond f$ skips over the entire interval $[u, v]$ and thus $\|st\diamond f\| = \|st\diamond [u, v]\|\geq F_{s, t}^{[u, v]} \geq \min\{\|su\| + \|ut\diamond f\|, \|sv\diamond f\| + \|vt\|, F_{s, t}^{[u, v]}\}$.
\end{proof}

\section{The tree partition lemma}
	Our novel DSO begins with the following lemma. This lemma also appeared in \cite{virgi12}.
	
	\begin{lem}[Tree partition\label{treepart}]
		Given a rooted tree $\mathcal{T}$, and any integer $2\leq k\leq n = |V(\mathcal{T})|$, there exists a subset of vertices $M\subseteq V(\mathcal{T})$, $|M| \leq 3k - 5$, such that after removing all vertices in $M$, the tree $\mathcal{T}$ is partitioned into sub-trees of size $\leq n/k$. We call every $u\in M$ an $M$-marked vertex, and $M$ a marked set. Plus, such $M$ can be computed in $O(n\log k)$ time.
	\end{lem}
	\begin{proof}
		We prove this claim by an induction on $k$.
		\begin{itemize}
			\item Basis $k = 2$.\\
			We argue that deleting one vertex from $\mathcal{T}$ can decompose the entire tree into sub-trees of size $\leq n/2$.
			
			First, we can compute the size of sub-tree rooted at every vertex by a traversal, which takes linear time. Then, consider the following procedure. To locate the one that needs to be deleted, start the search at root and travel down the tree. We will be met with two possible cases.
			\begin{enumerate}[(1)]
				\item If the sizes of all the sub-trees rooted immediately below our current position are $< n/2$, then we stop and return with the vertex at current position.
				\item Otherwise, we step onto the child whose corresponding sub-tree has $\geq n/2$ vertices.
			\end{enumerate}
			The procedure is guaranteed to terminate since the maximum sub-tree size rooted below the current position always decreases every time we step onto a child. Correctness is also ensured because during the process we maintain the invariant that if the current vertex is deleted from tree, the component containing the parent vertex is of size $\leq n/2$. Time complexity of this procedure is $O(n)$, because the searching visits every vertex no more than once.
			
			\item Induction for $k \geq 3$.\\
			Apply the lemma with parameter $\lceil\frac{k}{2}\rceil \in [2, k)$, and the tree $\mathcal{T}$ is broken into sub-trees of size $\leq n / \lceil\frac{k}{2}\rceil$. The total number of sub-trees whose sizes are $> n / (2\lceil\frac{k}{2}\rceil)$ is at most $2\lceil\frac{k}{2}\rceil - 1$. Applying the base case with parameter 2, we can further decompose all those ``larger'' sub-trees into sub-trees of size $\leq n / (2\lceil\frac{k}{2}\rceil) \leq n / k$ with no more than $2\lceil\frac{k}{2}\rceil - 1$ deletions. The total number of deletions would thus be $\leq 2\lceil\frac{k}{2}\rceil- 1 + 3\lceil\frac{k}{2}\rceil - 5 \leq 3k - 5$; the last inequality holds for all $k\geq 3$.
			
			Running time of this algorithm for general $k$ satisfies recursion $T(k) = T(\lceil k / 2\rceil) + O(n)$, which leads to $T(k) = O(n\log k)$.
		\end{itemize}
	\end{proof}

The high-level idea of our data structure is that we reduce the computation of an arbitrary $\|st\diamond f\|$ to a ``shorter'' $\|uv\diamond f\|$; here we say ``short'' in the sense that either $|uf|$ or $|fv|$ is small. The tree partition lemma helps us with the reduction. Basically, we apply the lemma \textbf{twice} with different parameters so that either $uf$ or $fv$ becomes ``short'' enough, and then we can directly retrieve the length of replacement path from storage.\\

For the rest of this paper, we say a replacement path $st\diamond f$ is $L$-short with respect to marked set $M$ and $T_s$ ($\widehat{T}_t$), if $t$ ($s$) and $f$ lie within the same sub-tree after we remove the entire marked vertex set $M$; here $M$ guarantees that each sub-tree has size $<L$. Sometimes we don't explicitly refer to the corresponding $M$ when discussing $L$-shortness, and nor do we specify which tree, $T_s$ or $\widehat{T}_t$, is being partitioned. These conditions are supposed to be easily learnt from context.

\section{Reducing to $\log^2 n$-short paths}
We devise an $O(n^2)$-space data structure that computes all non-$\log^2 n$-short paths in constant time.

	\subsection{Data structure \label{log2}}
	Just like in section \textbf{\ref{old}}, our DSO first pre-computes all values of $\|st\|$ and $|st|$, which accounts for $O(n^2)$ space. Then, for each $s\in V$, apply lemma \textbf{\ref{treepart}} in $T_s$ ($\widehat{T}_s$) to obtain a marked set $M_s$ ($\widehat{M}_s$) with parameter $\lceil n / (L-1)\rceil$, where $2\leq L\leq n$ is an integer to be set later. So $M_s$ ($\widehat{M}_s$) is of size $O(n / L)$, and the size of each sub-tree is $<L$. Consider the following structures. Note that we always build same structures in the reverse graph.\\
	
\begin{note}
The un-weighted distances between two adjacent marked vertices in $T_s$ ($\widehat{T}_s$) are $\leq L$.
\end{note}

	For any pair of $s, t$ such that $t\in M_s$, suppose we are met with $M_s$-marked vertices $u_1\rightarrow u_2\rightarrow\cdots\cdots\rightarrow u_k = t$ along the path $st$ in $T_s$. Our data structure consists of several parts.
	\begin{enumerate}[(i)]
		\item For each $k - 2^i\in [1, k-1]$, the value of $\|st\diamond u_{k - 2^i}\|$.
		\item For each $k - 2^i\in [1, k-2]$, the value of $\|st\diamond [u_{k - 2^i}, u_{k - 2^i + 1}]\|$.
		\item Let $v_l\rightarrow \cdots\cdots \rightarrow v_1$ be the sequence of all $\widehat{M}_t$-marked vertices along the path $st$. Then for each properly defined interval $[v_{l - 2^i}, u_{k - 2^j}]$ on the path $st$, store the value of $\|st\diamond [v_{l - 2^i}, u_{k - 2^j}]\|$.
		\item For each $f$ such that $|ft| \leq 2L$ or $|sf|\leq 2L$, the value of $\|st\diamond f\|$.\\

		From now on we drop the assumption that $t$ is $M_s$-marked. 		
		\item For every pair $(s, t)$ of different vertices, let $x$ be $t$'s nearest $M_s$-marked $T_s$-ancestor, and $y$ be $s$'s nearest $\widehat{M}_t$-marked $\widehat{T}_t$-ancestor. If intervals $(s, x]$ and $[y, t)$ intersect, then we pre-compute and store $\|st\diamond [y, x]\|$. Also, store addresses of $x, y$, if such ancestors exist.
		\item Build a tree upon all $M_s\cup \{s\}$'s vertices as follows. In this tree, $u$ is $v$'s parent if and only if in $T_s$ $u$ is $v$'s nearest ancestor that belongs to $M_s\cup \{s\}$. Then pre-compute and store the level-ancestor data structure of this tree. Note again that we also build similar structures for $\widehat{M}_s\cup\{s\}$ in the reverse graph.
	\end{enumerate}

	Conduct a simple space complexity analysis for each part of the data structure.
	\begin{enumerate}[(i)]
		\item takes up space $O(\frac{n^2\log n}{L})$ since we have $O(\log n)$ choices for the index $i$, and every $|M_s| = O(n / L)$.
		\item uses $O(\frac{n^2\log n}{L})$ for a similar reason in the previous part.
		\item demands $O(\frac{n^2\log^2 n}{L})$ space since we have $O(\log^2 n)$ choices for the pair of $(i, j)$.
		\item entails an $O(n^2)$ space consumption since each $M_s$-marked $t$ is associated with $O(L)$ entries, and there are $O(n / L)$ $M_s$-marked vertices $t$.
		\item induces $O(n^2)$ space complexity.
		\item takes $O(n^2 / L)$ total space, each tree of size $O(n / L)$.
	\end{enumerate}
	
	Therefore, the overall space complexity from (\romannumeral1) through (\romannumeral6) is equal to $O(\frac{n^2\log^2 n}{L} + n^2)$. Taking $L = \log^2 n$, it becomes $O(n^2)$.
	
	\subsection{Query algorithm}
	We prove the following reduction lemma in this sub-section.
	
	\begin{lem}
		The data structure specified in section \textbf{\ref{log2}} can compute $\|st\diamond f\|$ in $O(1)$ time if $st\diamond f$ is not $L$-short with respect to $M_s$ or $\widehat{M}_t$.
	\end{lem}
	\begin{proof}
		A constant time verification for $L$-shortness is easy: we check if $f$ lies below the nearest $M_s$-marked $T_s$-ancestor of $t$, and similarly if $f$ lies below the nearest $\widehat{M}_t$-marked $\widehat{T}_t$-ancestor of $s$.\\
		
		Firstly we argue that it is without loss of generality to assume that $t$ is $M_s$-marked. The reduction proceeds as follows.
		
		Let $x$ and $y$ be vertices defined as in (\romannumeral5). Since $st\diamond f$ is not $L$-short, $f\in (s, x]\cap [y, t)$, and thus $[y, x]$ is a properly defined interval on path $st$. By the \textbf{triple path lemma}, one has: $$\|st\diamond f\| = \min\{\|sx\diamond f\| + \|xt\|, \|sy\| + \|yt\diamond f\|, \|st\diamond [y, x]\|\}$$
		Here we use the interval admissible function $\|st\diamond [y, x]\|$. Noticing that the third term $\|st\diamond [y, x]\|$ is already covered in (\romannumeral5) from \textbf{\ref{log2}}, we are left with $\|sx\diamond f\|$ and $\|yt\diamond f\|$. By definition, $x$ is $M_s$-marked and $y$ is $\widehat{M}_t$-marked, and thus we complete our reduction.\\

		Let $u_1\rightarrow u_2\rightarrow \cdots\cdots \rightarrow u_p = t$ be the sequence of all $M_s$-marked vertices along $st$. We can assume that $p > 1$; otherwise $\|st\diamond f\|$ has already been computed in structure (\romannumeral4).
		
		It is not hard to find the interval $[u_a, u_{a+1})$ that contains $f$. On the one hand, $u_a$ is easily retrieved: if $f$ not $M_s$-marked, then $u_a$ is its nearest marked ancestor stored in (\romannumeral5); otherwise, $u_a = f$. On the other hand, $u_{a+1}$ can be found by querying the level-ancestor data structure (\romannumeral6) at node $t$ in tree rooted at $s$.
		
		Let $v_q\rightarrow v_{q-1}\rightarrow \cdots\cdots \rightarrow v_1$ be the sequence of all $\widehat{M}_t$-marked vertices on $st$. It is also safe to assume $q > 1$; otherwise, we have $|st|\leq 2L$ and then using (\romannumeral4) we can directly compute $\|st\diamond f\|$. Similar to the previous paragraph, locate the interval $(v_{b+1}, v_b]$ that includes $f$. We only need to consider the case when $a+1 < p$ and $b+1 < q$, since otherwise $\|st\diamond f\|$ can be directly retrieved from structure (\romannumeral4).

		Find maximum indices $i, j\geq 0$ such that $q - 2^i \geq b+1$, $p - 2^j \geq a+1$. Applying the \textbf{triple path lemma} with respect to $[v_{q - 2^i}, u_{p - 2^j}]$ in terms of interval admissible function, $\|st\diamond f\|$ must be the minimum among the following three distances.
		
		\begin{enumerate}[(1)]
			\item $\|st\diamond [v_{q-2^i}, u_{p - 2^j}]\|$.\\
			This value is directly retrievable from (\romannumeral3) in section \textbf{\ref{log2}}.
			\item $\|su_{p - 2^j}\diamond f\| + \|u_{p - 2^j}t\|$.\\
			Note that since we are interested in the minimum among (1)(2)(3) which gives us $\|st\diamond f\|$, we can substitute any value for (2) that lies in range $[\|st\diamond f\|, \|su_{p - 2^j}\diamond f\| + \|u_{p - 2^j}t\|]$.
			
			By definition of $j$, $u_{a + 2^j}$ lies between $u_{p - 2^j}$ and $t$, and hence we know that the concatenation of paths $su_{p - 2^j}\diamond f$ and $u_{p - 2^j}t$ passes through $u_{a + 2^j}$. Thus, it must be $$\|su_{p - 2^j}\diamond f\| + \|u_{p - 2^j}t\|\geq \|su_{a + 2^j}\diamond f\| + \|u_{a + 2^j}t\|\geq \|st\diamond f\|$$
			
			So instead of computing the original (2), we are actually calculating $\|su_{a + 2^j}\diamond f\| + \|u_{a + 2^j}t\|$.
			
			We focus on the case when $f \neq u_a$; the case where $f = u_a$ is easy in that we can directly query $\|su_{a + 2^j}\diamond u_a\|$ using (\romannumeral1).
			
			Applying the \textbf{triple path lemma} for a third time to $su_{a + 2^j}\diamond f$ and interval $[u_a, u_{a+1}]$, we further divide it into three cases.
			
			\begin{enumerate}[(a)]
				\item $\|su_{a + 2^j}\diamond [u_a, u_{a+1}]\| + \|u_{a + 2^j}t\|$.\\
				This can be computed by a single table lookup in (\romannumeral2).
				\item $\|su_{a+1}\diamond f\| + \|u_{a+1}t\|$.\\
				Since $u_{a+1}$ is $M_s$-marked, $\|su_{a+1}\diamond f\|$ is stored in structure (\romannumeral4), and thereby $\|su_{a+1}\diamond f\| + \|u_{a+1}t\|$ is computed effortlessly.
				\item $\|su_a\| + \|u_a t\diamond f\|$.\\
				If $u_a$ itself is $\widehat{M}_t$ marked, then (\romannumeral4) directly help us out since $\|u_at\diamond f\|$ is already pre-computed as $|u_af|\leq L$.

				Otherwise, suppose $v$ is $u_a$'s nearest $\widehat{M}_t$-marked ancestor in $T_s$ (if any). To locate such $v$, we can try to find the interval $(v_{c+1}, v_c]$ that contains $u_a$, in a similar fashion of finding intervals $[u_a, u_{a+1})$ and $(v_{b+1}, v_b]$; after that we assign $v\leftarrow v_{c+1}$.
				
				If such $v$ does not exist, then $s$ and $u_a$ lie in the same sub-tree of $\widehat{T}_t$ after removing $\widehat{M}_t$. Noticing that $|sf| < |su_{a+1}| = |su_a| + |u_a u_{a+1}| \leq 2L$, (\romannumeral4) can finish up $\|st\diamond f\|$ by a single table look-up. If $v$ exists, then by $\|su_a\| + \|u_a t\diamond f\|\geq \|sv\| + \|vt\diamond f\|\geq \|st\diamond f\|$, it suffices to compute $\|vt\diamond f\|$, which also has already been pre-computed in (\romannumeral4) due to $|vf| = |vu_a| + |u_a f|\leq 2L$.
			\end{enumerate}

			\item $\|sv_{q - 2^i}\| + \|v_{q - 2^i}t\diamond f\|$.\\
			The only difficult part is $\|v_{q - 2^i}t\diamond f\|$. Similar arguments in the previous case (2) will still work.
		\end{enumerate}
				
	\end{proof}

\section{An $\langle O(n^2\log\log n), O(1)\rangle$ construction}
In this section, we present an ordinary way of handling $L$-short paths, resulting in an $\langle O(n^2\log\log n), O(1)\rangle$ DSO. On a high level, we directly apply the sparse table construction as in section \textbf{\ref{old}}, which is from \cite{Demet08}. But since the sparse table only needs to cover $L$-short paths, the space requirement shrinks to $O(\log L) = O(\log\log n)$ for every pair of $s, t\in V$. Hence the total space complexity would be $O(n^2\log\log n)$.

\subsection{Data structure \label{loglog}}
For any pair of $s, t\in V$, besides $\|st\|, |st|$, build the following structures.

\begin{enumerate}[(i)]
\item For every $2^i\leq \min\{4L, |st|\}$, store $\|st\diamond (s\oplus 2^i)\|$ and $\|st\diamond (t\ominus 2^i)\|$.
\item For every $2^{i+1}\leq \min\{4L, |st|\}$, store $\|st\diamond [s\oplus 2^i, s\oplus 2^{i+1}]\|$ and $\|st\diamond [t\ominus 2^{i+1}, t\ominus 2^i]\|$.
\item Level ancestor data structures of $T_s$ and $\widehat{T}_s$.
\end{enumerate}

Since $L = \log^2 n$, the total space of this structure is equal to $O(n^2\log L) = O(n^2\log\log n)$. Note that this structure is basically identical to the one from section \textbf{\ref{old}}, except for the additional bound $4L$ on the power-of-two's $2^i$. Therefore, when $|st|\leq 4L$, $\|st\diamond f\|$ can be retrieved in $O(1)$ time according to the correctness guaranteed by \cite{Demet08}.

\subsection{Query algorithm}
We prove the following lemma, showing how our data structure covers all $L$-short paths.

\begin{lem}
For any $st\diamond f$ such that $|sf|\leq L$ or $|ft|\leq L$, $\|st\diamond f\|$ can be computed in $O(1)$ time by the data structure presented in section \textbf{\ref{loglog}}.
\end{lem}
\begin{proof}
By the symmetry of our data structure, we can assume without loss of generality that $|sf|\leq L$. It is also convenient to suppose that $|sf|$ is not a power of 2, otherwise $\|st\diamond f\|$ is directly retrievable from (\romannumeral1) in \textbf{\ref{loglog}}. If $|st|\leq 4L$, then query algorithm in section \textbf{\ref{old}} can answer $\|st\diamond f\|$ in constant time. Now let us assume $|st| > 4L$. Let $i$ be the largest index such that $2^i < |sf|$. Apply the \textbf{triple path lemma} for replacement path $st\diamond f$ and interval $[s\oplus 2^i, s\oplus 2L]$ and consider the following three distances. Note that since $2^i < |sf| \leq L < 2L$, $f$ lies between $s\oplus 2^i$ and $s\oplus 2L$.
\begin{enumerate}[(1)]
\item $\|s(s\oplus 2^i)\| + \|(s\oplus 2^i)t\diamond f\|$.\\
By definition of $i$, $f\ominus 2^i$ lies between $s$ and $s\oplus 2^i$. Then we have $$\|s(s\oplus 2^i)\| + \|(s\oplus 2^i)t\diamond f\|\geq \|s(f\ominus 2^i)\| + \|(f\ominus 2^i)t\diamond f\|\geq \|st\diamond f\|$$

Thereby it is legal to substitute $\|s(f\ominus 2^i)\| + \|(f\ominus 2^i)t\diamond f\|$ for (1). According to the construction in \textbf{\ref{loglog}}, both terms can be retrieved from memory in constant time.
\item $\|s(s\oplus 2L)\diamond f\| + \|(s\oplus 2L)t\|$.\\
The former term is computable in constant time because $|s(s\oplus 2L)| = 2L < 4L$.
\item $\|st\diamond [s\oplus 2^i, s\oplus 2L]\|$.\\
Since $2^{i+1} < 2|sf| \leq 2L$, $st\diamond f$ actually skips over the interval $f\in [s\oplus 2^i, s\oplus 2^{i+1}]$. Therefore, $$\|st\diamond [s\oplus 2^i, s\oplus 2L]\| \geq \|st\diamond [s\oplus 2^i, s\oplus 2^{i+1}]\|\geq \|st\diamond f\|$$

Substituting $\|st\diamond [s\oplus 2^i, s\oplus 2^{i+1}]\|$ for (3), we retrieve it via a single table look-up in (\romannumeral2) in constant time.
\end{enumerate}
\end{proof}

By definition, every $L$-short path $st\diamond f$ satisfies $|sf|\leq L$ or $|ft|\leq L$. So by the above lemma, data structure introduced in section \textbf{\ref{loglog}} answers every $L$-short path query. Together with the structures in section \textbf{\ref{log2}}, it makes an $\langle O(n^2\log\log n), O(1)\rangle$ DSO.

\section{Two-level partition}
In this section, we obtain an $\langle O(n^2), O(1)\rangle$ construction of DSO. The high-level idea is that we further partition every sub-tree into even smaller ones, and then we apply a tabulation (``Four Russians'') technique to store all answers. More specifically, we apply the tree-partitioning lemma for the second time and break each SSSP tree into sub-trees of size $\leq \log\log^2 n$. Then we devise data structures to reduce $\log^2 n$-short paths to $\log\log^2 n$-short paths. Finally, the tabulation technique kicks in when it comes to $\log\log^2 n$-short paths.

\subsection{Data structure \label{part2}}
Let $L^\prime \leq L$ be a parameter to be set later. For each SSSP tree $T_s$ ($\widehat{T}_s$), compute its tree partition with parameter $\lceil n / (L^\prime - 1)\rceil$ by Lemma~\ref{treepart}, and let $M^\prime_s$ ($\widehat{M^\prime}_s$) be the corresponding marked set. For any $L$-short path $st\diamond f$, without loss of generality assume that $t, f$ lie in the same sub-tree of $T_s$ after removing $M_s$, and let $r$ be the root of this sub-tree. Note that similar structures are also built for the case when $s, f$ lie in the same sub-tree of $\widehat{T}_t$.

\begin{enumerate}[(i)]
	\item If $t$ is not $M_s^{\prime}$-marked.\\
	Let $u$ be $t$'s nearest $M_{s}^{\prime}$-marked ancestor below $r$ (if such $u$ exists), and store the value of $\|st\diamond [r, u]\|$.
	\item If $t$ is $M_s^{\prime}$-marked.\\
	Let $u_1\rightarrow u_2\rightarrow \cdots\cdots \rightarrow u_k = t$ be the sequence of all $M_{s}^{\prime}$-marked ancestor along the directed path $rt$. Note that $k < L = O(\log^2 n)$. Then for each $k - 2^i \in [1, k-1]$, store the value of $\|st\diamond [r, u_{k - 2^i}]\|$. After that, for each $f\in [u_{k-1}, u_k)$ (define $u_0 = r$), store the value of $\|st\diamond f\|$.
	\item Build upon the marked set $M^\prime_s$ all structures from (\romannumeral1) to (\romannumeral6) in section \textbf{\ref{log2}}. The only difference is that we impose an additional constraint that $|st|\leq L$ on structures (\romannumeral1) through (\romannumeral3). It is not hard to verify that the space complexity of this part becomes $O(\frac{n^2\log^2 L}{L^\prime} + n^2) = O(\frac{n^2\log\log^2 n}{L^\prime} + n^2)$. So if $st\diamond f$ is non-$L^\prime$-short, with $|st|\leq L$, then applying lemma \textbf{\ref{log2}}, $\|st\diamond f\|$ can be answered in constant time.
\end{enumerate}

Note that the space complexity of (\romannumeral1) and (\romannumeral2) in section \textbf{\ref{part2}} is equal to $O(\frac{n^2\log\log n}{L^\prime} + n^2)$. Together with (\romannumeral3), the overall space complexity of the data structure is $O(n^2 + \frac{n^2\log\log n}{L^\prime} + \frac{n^2\log\log^2 n}{L^\prime})$. Taking $L^\prime = \log\log^2 n$, the space becomes $O(n^2)$.

\subsection{Reduction algorithm}
We prove the following lemma.

	\begin{lem}
		Given an $L$-short replacement path $st\diamond f$, the data structure in sections \textbf{\ref{part2}} and \textbf{\ref{log2}} can reduce $\|st\diamond f\|$ to a constant number of $\|uv\diamond f\|$'s, where $uv\diamond f$'s are $L^\prime$-short with respect to $M^\prime_u$ or $\widehat{M^\prime}_v$.
	\end{lem}
	\begin{proof}
		Constant time verification of $L^\prime$-shortness is easy. So we only need to consider when $st\diamond f$ is not $L^\prime$-short.
		
			We first argue that we will only need to concentrate on situations where $t$ is $M^\prime_s$-marked. Here is the reduction. If $t$ is not $M^\prime_s$-marked, locate $t$'s nearest $M^\prime_s$-marked ancestor $u$ below $f$. Vertex $u$ must exist, since otherwise $st\diamond f$ itself would become an $L^\prime$-short path. Recall that $r$ is the root of the sub-tree after removing $M_s$, and since $st\diamond f$ is $L$-short, $f$ must be after $r$ in $st$.
			
			We apply the \textbf{triple path lemma} on $st\diamond f$ and interval $[r, u]$.
			
			\begin{enumerate}[(1)]
				\item $\|sr\| + \|rt\diamond f\|$.\\
				In this case, we only need to consider $\|rt\diamond f\|$.
				
				If $rt\diamond f$ is not $L^\prime$-short, then since $|rt|$ is guaranteed to be $\leq L$, similar to the proof of lemma \textbf{\ref{log2}}, $\|rt\diamond f\|$ can be computed in $O(1)$ time by (\romannumeral3).
				
				If $rt\diamond f$ is $L^\prime$-short, then we deliver $rt\diamond f$ as one result of the reductions.
				
				\begin{note}
					Here $rt\diamond f$'s $L^\prime$-shortness is defined by $M_r^\prime$ or $\widehat{M}_t^\prime$, with respect to tree $T_r$ or $\widehat{T}_t$.
				\end{note}
				
				\item $\|st\diamond [r, u]\|$.\\
				We retrieve this distance directly from (\romannumeral1) in section \textbf{\ref{part2}}.
				\item $\|su\diamond f\| + \|ut\|$.\\
				The reduction ends here. Next we will continue to work out $\|su\diamond f\|$.
			\end{enumerate}
			
			Now we can assume without loss of generality that $t$ is $M^\prime_s$-marked. Let $u_1\rightarrow u_2\rightarrow\cdots\cdots \rightarrow u_p = t$ be the sequence of all $M_s^\prime$-marked vertices along the path $rt$. Suppose $f\in [u_a, u_{a+1})$, $a\geq 0$ ($u_0 = r$). If $a+1 = p$, then (\romannumeral2) enables us to directly query $\|st\diamond f\|$. Otherwise, let $i$ be the maximum index such that $p - 2^i \geq a+1$. Invoke again the \textbf{triple path lemma} on $st\diamond f$ and interval $[r, u_{p - 2^i}]$.
			
			\begin{enumerate}[(1)]
				\item $\|sr\| + \|rt\diamond f\|$.\\
				This can be solved in a similar fashion as in the previous case (1).
				\item $\|st\diamond [r, u_{p - 2^i}]\|$.\\
				(\romannumeral2) gives us $\|st\diamond [r, u_{p - 2^i}]\|$ in constant time.
				\item $\|su_{p - 2^i}\diamond f\| + \|u_{p - 2^i}t\|$.
				Such a replacement path is guaranteed to pass through $u_{a + 2^i+1}$. Thereby, $$\|su_{p - 2^i}\diamond f\| + \|u_{p - 2^i}t\|\geq \|su_{a + 2^i + 1}\diamond f\| + \|u_{a + 2^i + 1}t\|\geq \|st\diamond f\|$$
				So it suffices to compute $\|su_{a + 2^i + 1}\diamond f\| + \|u_{a + 2^i + 1}t\|$. Invoking the \textbf{triple path lemma} on $su_{a + 2^i + 1}\diamond f$ and interval $[r, u_{a+1}]$, we divide it into the following three cases.
				
				\begin{enumerate}[(a)]
					\item $\|su_{a+1}\diamond f\| + \|u_{a+1}u_{a + 2^i + 1}\|$.\\
					$\|su_{a+1}\diamond f\|$ can be retrieved from (\romannumeral2).
					\item $\|sr\| + \|ru_{a + 2^i + 1}\diamond f\|$.\\
					If $ru_{a + 2^i +1}\diamond f$ is not $L^\prime$-short, then by $|ru_{a + 2^i + 1}| \leq L$, $\|ru_{a + 2^i + 1}\diamond f\|$ can be calculated in $O(1)$ time by (\romannumeral3) in section \textbf{\ref{part2}}, similar to the proof of lemma \textbf{\ref{log2}}.
					
					If $ru_{a + 2^i +1}\diamond f$ is $L^\prime$-short, then we deliver it as one result of the reductions.
					
					\begin{note}
						Here $ru_{a + 2^i + 1}\diamond f$'s $L^\prime$-shortness is defined by $M_r^\prime$ or $\widehat{M}_{u_{a + 2^i + 1}}^\prime$, with respect to tree $T_r$ or $\widehat{T}_{u_{a + 2^i +1}}$.
					\end{note}

					\item $\|su_{a + 2^i + 1}\diamond [r, u_{a+1}]\|$.\\
					$\|su_{a + 2^i + 1}\diamond [r, u_{a+1}]\|$ can be found from memory storage of (\romannumeral2) in section \textbf{\ref{part2}}.
				\end{enumerate}

			\end{enumerate}
			
	\end{proof}

\subsection{Tabulation}
In this sub-section, we handle all $L^\prime$-short paths. Recall that the notation $\Delta_{s, t, f}, \nabla_{s, t, f}$ refers to divergence and convergence of replacement path $st\diamond f$. We will simply write $\Delta, \nabla$ because indices $s, t, f$ can be learnt from context.

Let $st\diamond f$ be an $L^\prime$-short path; without loss of generality assume that $t, f$ lie in the same sub-tree, the corresponding marked set being $M_s^\prime$. One observation is that we only need to focus on cases where $|s\Delta| \leq L$: if the divergence comes after $s\oplus L$, then it admits the decomposition $\|st\diamond f\| = \|su\| + \|ut\diamond f\|$, $u$ being $s$'s nearest $\widehat{M}_t$-marked ancestor in tree $\widehat{T}_t$. Since $|ft| \leq L^\prime < L < 2L$, $\|ut\diamond f\|$ can be found in (\romannumeral4) from section \textbf{\ref{log2}}.

For each sub-tree $T$ partitioned by marked set $M^\prime_s$, we in-order sort all its vertices. The aggregate divergence / convergence information within this sub-tree can be summarized as an $L^\prime \times L^\prime$ matrix, each element being a pair $(|s\Delta|, |\nabla t|)$ corresponding to a replacement path $st\diamond f$, $\forall t, f\in V(T)$. Since we only consider the case when $|s\Delta| \leq L$, the total number of choices for this matrix is no more than $(L \cdot L^\prime)^{(L^\prime)^2} < L^{2(L^\prime)^2} = 2^{4\log\log^5 n} < O(2^{\log n}) = O(n)$.

Construct an indexable table of all possible configurations of such matrices. The space of this table is $\leq O(n\cdot (L^\prime)^2) = o(n^{1.1})$. Then associate each sub-tree with an index of its corresponding matrix in the table, which demands a storage of $O(n / L^\prime)$ indices, totalling $o(n)$ space for every $s$. Thus the overall space complexity associated with tabulation is $o(n^2)$.

Now the $L^\prime$-short $\|st\diamond f\|$ can be computed effortlessly. After indexing the corresponding matrix in the table, we can extract $(|s\Delta|, |\nabla t|)$ directly from this matrix, and then recover $\Delta, \nabla$ from level-ancestor data structures. Finally, decompose the replacement path as $\|st\diamond f\| = \|s\Delta\| + \|\Delta\nabla\diamond (\Delta, \nabla)\| + \|\nabla t\|$. Noticing that $\Delta\nabla\diamond f = \Delta\nabla\diamond (\Delta, \nabla)$, thereby the value of $\|\Delta\nabla\diamond f\|$ is equal to any admissible function value $F_{\Delta, \nabla}^{[\Delta\oplus 1, \nabla\ominus 1]}$. Hence, storing a $\|uv\diamond [u\oplus 1, v\ominus 1]\|$ for every pair of $u, v$ will suffice for querying $\|st\diamond f\|$ once divergence and convergence vertices are known.

\section{Concluding remarks}
It is also convenient to slightly extend our data structure to retrieve the entire replacement path in $O(|st\diamond f|)$ time, for any query $(s, t, f)$: for each entry in our DSO of the form $\|st\diamond A\|$, $A$ being a vertex set, store the first edge of the path $st\diamond A$. This adjustment does not affect our asymptotic space complexity. In this way, for any $(s, t, f)$, the DSO can answer not only $\|st\diamond f\|$, but also the first edge, say $(s, s^\prime)\in E$, of replacement path $st\diamond f$. To recover the entire $st\diamond f$, we recur on $s^\prime t\diamond f$. Since each iteration reveals a single edge, the total time would be $O(|st\diamond f|)$.

So far we have devised an $\langle O(n^2), O(1)\rangle$ DSO. Clearly both of the space complexity and query efficiency have reached asymptotic optima; also its preprocessing time is $\tilde{O}(mn)$ (see appendix), which is nearly optimal.


\appendix

\section{Preprocessing algorithm}
In this section, we present a preprocessing algorithm that runs in $\tilde{O}(mn)$ time. Basically, computing replacement path that skips an entire interval is hard. To eschew such hardness, we employ the idea of \textsl{admissible functions} \cite{Bernst08} and change replacement paths of the form $st\diamond [u, v]$ to certain $st\diamond f$ ones without harming the correctness of query algorithms. \emph{Indeed, the data structure we constructs here is not exactly the one we have described, but an equivalent one.} Then, to build our DSO, we first apply the randomized preprocessing algorithm from \cite{Bernst09} that runs in time $O(mn\log n + n^2\log^2 n)$ to obtain a less space-efficient DSO. Finally, using this intermediate DSO we can compute the $\langle O(n^2), O(1)\rangle$ DSO in $\tilde{O}(mn)$ time.

\begin{note}
	Although invoking \cite{Bernst09}'s algorithm as a black-box may introduce randomness, rest are entirely deterministic. So if we apply the deterministic but slightly less efficient preprocessing algorithm that also appeared in \cite{Bernst09}, our initializing algorithm can then become deterministic, still with $\tilde{O}(mn)$ time bound.
\end{note}

\subsection{Preliminary}
Values of $\|st\|$ and $|st|$ can be handled by invoking standard all-pairs shortest paths (APSP) algorithms,  which will not be our bottleneck. Computing Interval admissible function values is time consuming. In our preprocessing algorithm, instead of computing entry values of the form $\|st\diamond [u,v]\|$, we compute $\max_{f\in[u,v]}\{\|st\diamond f\|\}$; namely we replace every interval admissible function value demanded by our DSO with its bottleneck counterpart. For example, in (\romannumeral2) from \ref{log2}, what we actually compute is $\max_{f\in[u_{k- 2^i} ,u_{k- 2^i +1}]} {\|st\diamond f\|}$. All the query / reduction algorithms need to be adapted accordingly. This is not hard because the triple path lemma holds for arbitrary admissible functions.

Hence, our $\langle O(n^2), O(1)\rangle$ DSO is mostly composed of entries of the form $\max_{f\in[u, v]}\|st\diamond [u, v]\|$. Once all such entries are computed, other parts (say, tabulation) become easy. Therefore we only need to deal with entries that take the form of $\max_{f\in[u, v]}\|st\diamond [u, v]\|$.

The idea is to build an oracle that answers any query $\max_{f\in[u, v]}\|st\diamond f\|$ in $O(\log n)$ time; note that this oracle is more powerful than ordinary DSO. As we will see, this stronger oracle can be constructed in time $\tilde{O}(mn)$. Firstly we invoke \cite{Bernst09}'s algorithm that builds an $\langle O(n^2\log n), O(1)\rangle$ DSO in $\tilde{O}(mn)$ time, by which queries of $\|st\diamond f\|$ can be computed in constant time. However, we cannot afford to peek every value of $\|st\diamond f\|, f\in[u, v]$ in order to obtain $\max_{f\in [u, v]}\|st\diamond f\|$, since otherwise the preprocessing time would be blown up to an undesirable $\Omega(n^3)$. Fortunately, an important sub-routine from \cite{Bernst09}'s algorithm directly helps us out: instead of grid-searching along the interval $[u, v]$, one can perform a certain kind of binary search on $[u, v]$ to identify the bottleneck $f$ that maximizes $\|st\diamond f\|$.

We inherit some notations from \cite{Bernst09}. For any interval $[u, v]$ on the shortest path $st$, as well as a presumably failed vertex $f\in [u, v]$, define:
$$MTC(s, t, f, u, v) \overset{\Delta}{=} \min\{\|su\| + \|ut\diamond f\|, \|sv\diamond f\| + \|vt\|\}$$

So by the triple path lemma, $\|st\diamond f\| = \min\{MTC(s, t, f, u, v), \|st\diamond [u, v]\|\}$. A key observation from \cite{Bernst09} is that, the bottleneck vertex $f$ always maximizes the first term $MTC(s, t, f, u, v)$, because the second term $\|st\diamond [u, v]\|$ is independent of $f$. Therefore, in order to locate bottleneck vertex $f$, it is equivalent to maximizing $MTC(s, t, f, u, v)$.

The following lemma will also be useful.
\begin{lem}[\cite{Bernst09}]\label{skip}
If $\|su\| + \|ut\diamond f\| > \|sv\diamond f\| + \|vt\|$, then $sv\diamond f$ skips over $u$.
\end{lem}

\subsection{Data structure}
In this section, we introduce the structures of our stronger oracle that answers any queries of form $\max_{f\in[u, v]}\|st\diamond f\|$.

Generally speaking, we apply a scaling technique that deals with intervals $[u, v]$ of different sizes. For any $1\leq i\leq \log_2 n$, For every reverse SSSP tree $\widehat{T}_t$, apply the tree partition lemma \ref{treepart} with $L = 2^i$ and obtain a marked set $B_t^i$, $|B_t^i| = O(n / 2^i)$. For any marked vertex $s\in B_t^i$, construct an array of length $2^{i+2}$, for each $f\in (s, s\oplus 2^{i+1}]$ fill in the $|sf|^{th}$ entry with value $\|st\diamond f\|$ (again note that this value is retrievable in constant time via the intermediate DSO from \cite{Bernst09}). If $|st|\leq 2^{i+2}$, we simply cover the entire $st$. After that, build a \textsl{range maximum query} data (RMQ) structure on this array, which allows us to locate the bottleneck vertex in constant time on any interval within $(s, s\oplus 2^{i+2}]$. Symmetrically, we compute set marked $\widehat{B}_s^i$ associated with each SSSP tree $T_s$ and build a similar RMQ.

For each $i$, the construction time is only $O(n\cdot n/2^i\cdot 2^{i+2}) = O(n^2)$. So the overall construction time is $\tilde{O}(n^2)$.

\subsection{Query algorithm}
In this section, we show how bottleneck vertices can be efficiently computed using data structures in the previous section.

Suppose we wish to compute a bottleneck vertex in $[u, v]$ on $st$. Assume $2^{i+1}\leq |uv| < 2^{i+2}$, $i>0$. Since $|uv| \geq 2^{i+1}$, both of $\widehat{B}_s^i\cap ((u+v)/2, v]$ and $B_t^i\cap [u, (u+v)/2]$ are non-empty. Using level ancestor data structures we can locate a vertex $u^\prime\in B_t^i\cap [u, (u+v)/2]$ and $v^\prime\in \widehat{B}_s^i\cap ((u+v)/2, v]$.

To finish up, now we can directly apply the poly-log time sub-routine from \cite{Bernst09} to locate the bottleneck vertex on interval $[u^\prime, v]$, as described in a recursive procedure \textsc{FindBot} below, where $s, t, u^\prime, v$ are deemed global parameters, and the input $[x, y]\subseteq [u^\prime, v]$ is a sub-interval. The bottleneck vertex on interval $[u, v^\prime]$ can be computed in a symmetrical way.

\begin{algorithm}
\caption{\textsc{FindBot$([x, y])$} from \protect\cite{Bernst09}}
\If{$|xy| \leq 2$}{
	compute the botttleneck vertex in the trivial way and \Return\;
}
$q\leftarrow (x+ y)/2$, namely the mid-vertex on $xy$\;
$b\leftarrow$ the bottleneck vertex with respect to interval $[q, y]$ on $u^\prime t$, using RMQ\;
\If{$\|su^\prime\| + \|u^\prime t\diamond b\|\leq \|sv\diamond b\| + \|vt\|$}{
	$w\leftarrow$ \textsc{FindBot}$([x, q))$\;
	\Return $\arg\max_{b, w}\{MTC(s, t, b, u^\prime, v), MTC(s, t, w, u^\prime, v)\}$\;
}\Else {
	\Return \textsc{FindBot}$([q, y])$\;
}
\end{algorithm}

Clearly the running time is $O(\log n)$. Therefore, our query algorithm is $O(\log n)$. Correctness follows from exactly the same arguments as in \cite{Bernst09}, but for completeness we restate the proof below.

We show that \textsc{FindBot}$([u^\prime, v])$ finds a vertex $f$ that maximizes $MTC(s, t, f, u^\prime, v)$. Consider an arbitrary level of recursion with input $[x, y]$. Let's focus on line-5. If the branching condition holds true, then it suffices to prove that the bottleneck vertex with respect to interval $[q, y]$ is $b$ itself. In fact, for any $f\in[q, y]$, we have:
$$MTC(s, t, f, u^\prime, v)\leq \|su^\prime\| + \|u^\prime t\diamond f\|\leq \|su^\prime\| + \|u^\prime t\diamond b\| = MTC(s, t, b, u^\prime, v)$$
The second inequality is guaranteed by the function of RMQ.

Alternatively, if the branching condition is false, namely $\|su^\prime\| + \|u^\prime t\diamond b\| > \|sv\diamond b\| + \|vt\|$, we argue that the bottleneck vertex cannot be in interval $[x, q)$. For any $f\in [x, q)$, we have:
$$MTC(s, t, f, u^\prime, v)\leq \|sv\diamond f\| + \|vt\|\leq \|sv\diamond b\| + \|vt\| = MTC(s, t, b, u^\prime, v)$$
The second inequality holds because according to lemma \ref{skip}, $sv\diamond b$ skips over $u^\prime$, thus the entire interval $[u^\prime, q]$ which subsumes $f$, and this yields $\|sv\diamond f\|\leq \|sv\diamond b\|$.

\end{document}